\documentclass[journal]{IEEEtran}
\usepackage{amssymb}
\usepackage{amsmath}
\usepackage{cite}
\usepackage{algorithm}
\usepackage{algorithmic}
\usepackage[dvips]{graphicx}
\usepackage{stfloats}
\begin{document}
%
\title{Capacity Performance of Relay Beamformings for MIMO Multi-Relay Networks with Imperfect $\mathcal{R}$-$\mathcal{D}$ CSI at Relays}

\author{Zijian Wang, Wen Chen, Feifei Gao, and Jun Li, \em{Members,
IEEE}
\thanks{Copyright (c) 2011 IEEE. Personal use of this material is permitted. However,
permission to use this material for any other purposes must be obtained from the IEEE
by sending a request to pubs-permissions@ieee.org.}
\thanks{Z. Wang and W. Chen are with Wireless Network Transmission Laboratory, Shanghai
        Jiao Tong University, China. Z. Wang is also with ISN SKL, Xidian University, and
        W. Chen is also with SEU SKL for mobile communications, e-mail: \{wangzijian1786; wenchen\}@sjtu.edu.cn.}
\thanks{F. Gao is with the Department of Automation, Tsinghua University,
Beijing 100084, China and is also with the School of Engineering and
Science, Jacobs University, Bremen, 28759, Germany, Email:
feifeigao@ieee.org.}
\thanks{J. Li is with the School of Electrical
Engineering and Telecommunications, University of New South Wales,
Australian. Email: jun.li@unsw.edu.au.}
\thanks{This work is supported by NSF China \#60972031, by SEU SKL project
\#W200907, by ISN project \#ISN11-01, by Huawei Funding
\#YJCB2009024WL and \#YJCB2008048WL,
and by National 973 project \#2009CB824900.}
}

\markboth{IEEE Transactions on Vehicular Technology.}{} \maketitle

\begin{abstract}
In this paper, we consider a dual-hop Multiple Input Multiple Output
(MIMO) wireless relay network in the presence of imperfect channel
state information  (CSI), in which a source-destination pair both
equipped with multiple antennas communicates through a large number
of half-duplex amplify-and-forward (AF) relay terminals. We
investigate the performance of three linear beamforming schemes when
the CSI of relay-to-destination ($\mathcal{R}$-$\mathcal{D}$) link is not perfect at the
relay nodes. The three efficient linear beamforming schemes are
based on the matched-filter (MF), zero-forcing (ZF) precoding and
regularized zero-forcing (RZF) precoding techniques, which utilize
the CSI of both $\mathcal{S}$-$\mathcal{R}$ channel and $\mathcal{R}$-$\mathcal{D}$ channel at
the relay nodes. By modeling the $\mathcal{R}$-$\mathcal{D}$ CSI error at the relay nodes as
independent complex Gaussian random variables, we derive the ergodic capacities of the three beamformers in terms of instantaneous SNR. Using Law of Large Number, we obtain the asymptotic capacities, upon which the optimized MF-RZF is derived. Simulation results show that the asymptotic capacities match with the respective ergodic capacities very well. Analysis and simulation results demonstrate that the optimized MF-RZF outperforms MF and MF-ZF for any power of $\mathcal{R}$-$\mathcal{D}$ CSI error.
\end{abstract}

\begin{keywords}
MIMO relay, capacity, beamforming, channel state information.
\end{keywords}

\section{Introduction}
Relay communications can extend the coverage of wireless networks
and improve spatial diversity of cooperative systems.
 Meanwhile, MIMO
technique is well verified to provide significant improvement in the
spectral efficiency and link reliability because of the multiplexing
and diversity gains~\cite{1,2}. Combining the relaying and MIMO
techniques can make use of  both advantages to increase the data
rate in the cellular edge and extend the network coverage.

MIMO relay networks have been extensively investigated in
~\cite{3,5,7,31,32,33}. In addition MIMO multi-relay networks have been studied  in~\cite{4,6,8,9}. In~\cite{4}, the authors show that the
corresponding  network capacity scales as $C=(M/2) \log(K)+O(1)$,
where $M$ is the number of antennas at the source and
$K\rightarrow\infty$ is the number of relays. The authors also propose a
simple protocol to achieve the upper bound as $K\rightarrow\infty$ when perfect channel state informations (CSIs) of both source-to-relay ($\mathcal{S}$-$\mathcal{R}$) and relay-to-destination ($\mathcal{R}$-$\mathcal{D}$) channels are available at the relay nodes. When CSIs are not available at the relays, a simple AF beamforming protocol is proposed at the relays, but the distributed array gain is not obtained.
In~\cite{6}, a linear relaying scheme based on minimum mean square error (MMSE) fulfilling the target SNRs on
different substreams is proposed and the power-efficient relaying
strategy is derived in closed form for a MIMO multi-relay network. In~\cite{8,9}, the authors
design three relay beamforming schemes based on matrix
triangularization which have superiority over the conventional
zero-forcing (ZF) and amplify-and-forward (AF) beamformers. The proposed beamforming scheme can both fulfill intranode gain and distributed array gain.

However, most of the works only consider  perfect channel state
information (CSI) to design beamformers at the relays or successive
interference cancelation (SIC) matrices at the destination. For the
multi-relay networks, imperfect CSI of  $\mathcal{R}$-$\mathcal{D}$
channel is a practical consideration~\cite{4}. Especially, knowledge
for the CSI of $\mathcal{R}$-$\mathcal{D}$ channels at relays will
result in large delay and significant training overhead, because the
CSI of $\mathcal{R}$-$\mathcal{D}$ channels at the relays are
obtained through feedback links to multiple relays~\cite{40}.

For the works on imperfect CSI,
the ergodic capacity and BER performance of MIMO with imperfect CSI is considered in~\cite{14,16,19}.
In~\cite{14}, the authors investigated lower and upper bounds
of mutual information under CSI error. In~\cite{16}, the authors
studied BER performance of MIMO system under combined beamforming
and maximal ratio combining (MRC) with imperfect CSI. In~\cite{19},
bit error probability (BEP) is analyzed based on Taylor
approximation.
Some optimization problem has been investigated with
imperfect CSI in~\cite{15,17,18,20,21}. In~\cite{15}, the authors
maximize a lower bound of capacity by optimally configuring the
number of antennas with imperfect CSI. In~\cite{18}, assuming only
imperfect CSI at the relay, optimization problem of maximizing upper
bound of mutual information is presented and solved. In~\cite{21}, the authors studied
the trade-off between accuracy of channel estimation and data
transmission, and show that the optimal number of training symbols
is equal to the number of transmit antennas. In~\cite{34}, the authors investigate the effects of channel estimation error on the receiver of MIMO AF two-way relaying.

Recently,  two efficient relay-beamformers for the dual-hop MIMO
multi-relay networks have been presented in~\cite{10}, which are
based on matched filter (MF) and regularized zero-forcing (RZF), and
utilize QR decomposition (QRD) of the effective system channel
matrix at the destination node~\cite{13}. The beamformers at the
relay nodes can exploit the distributed array gain by diagonalizing
both the $\mathcal{S}$-$\mathcal{R}$ and $\mathcal{R}$-$\mathcal{D}$
channels. The QRD can exploit the intranode array gain by SIC
detection. These two beamforming schemes not only have advantageous
performance than that of the conventional schemes like QR-P-QR or
QR-P-ZF in~\cite{9}, but also have lower complexity because they
only need one QR decomposition at destination.  However, such
advantageous performances are based on perfect CSI, and the
imperfect CSIs of $\mathcal{R}$-$\mathcal{D}$ are not considered. It
is worth to know the capacity performances of these efficient
beamforming schemes and the validity of the scaling law in~\cite{4},
when the imperfect $\mathcal{R}$-$\mathcal{D}$ CSI presents at relays. In
addition, the asymptotic capacities for these beamforming schemes
and the optimal regularizing factor in the MF-RZF beamforming are
not derived in~\cite{10}.

Inspired by the works on  imperfect CSI and~\cite{10}, in this
paper, we investigate the performance of three efficient beamforming
schemes for dual-hop MIMO relay networks under the condition of
imperfect $\mathcal{R}$-$\mathcal{D}$ CSI at relays. The three
beamforming schemes are based on matched filter (MF), zero-forcing
(ZF) precoding and regularized zero forcing (RZF) precoding
techniques. We first derive the ergodic capacities in terms of the
instantaneous CSIs of $\mathcal{S}$-$\mathcal{R}$ and
$\mathcal{R}$-$\mathcal{D}$. Using Law of Large Number, we obtain
the asymtotic capacities for the three beamformers. Based on the
asymptotic capacity of MF-RZF, we derive the optimal regularizing
factor. Simulation results show that the asymptotic capacities match
with the ergodic capacities very well. Analysis and simulations
demonstrate that the capacity of MF-ZF drops fast when
$\mathcal{R}$-$\mathcal{D}$ CSI error increases. We observe that
MF-RZF always outperforms MF as in~\cite{10} when perfect CSI is
available at relays, while MF-RZF underperforms  MF when
$\mathcal{R}$-$\mathcal{D}$ CSI error is in presence. However, the
optimized MF-RZF always outperforms MF for any power of CSI error. The
ceiling effect of capacity is also discussed in this paper.

The remainder of this paper is organized as follows.  In Section
\uppercase \expandafter {\romannumeral 2}, the system model of a
dual-hop MIMO multi-relay network is introduced. In Section \uppercase
\expandafter {\romannumeral 3}, we briefly explain the three beamforming
schemes and QR decomposition. In Section \uppercase \expandafter
{\romannumeral 4}, we derive the instantaneous SNR on each antenna link at the
destination with $\mathcal{R}$-$\mathcal{D}$ CSI error at each relay. Using Law of Large Number, we obtain the asymptotic
capacities in Section \uppercase \expandafter
{\romannumeral 5}. Section \uppercase \expandafter
{\romannumeral 6} devotes to simulation results followed by
conclusion in Section \uppercase \expandafter {\romannumeral 7}.

\begin{figure}[!t]
\centering
\includegraphics[width=3.5in]{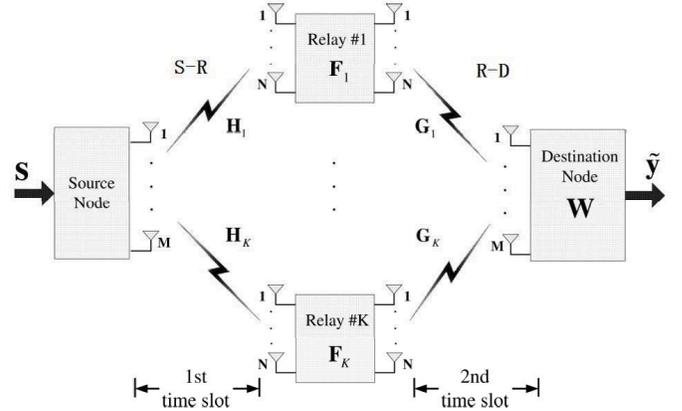}
\caption{System model of a dual-hop MIMO multi-relay network with relay
beamforming and SIC at the destination.}
\end{figure}
In this paper, boldface lowercase letter and boldface uppercase
letter represent vectors and matrices, respectively. Notations $\left( {\bf{A}} \right)_{i}$ and  $\left( {\bf{A}} \right)_{i,j}$
denote the $i$-th row and $(i,j)$-th entry of the matrix
${\bf{A}}$. Notations $\mathrm{tr}(\cdot)$ and $(\cdot)^H$ denote
trace and conjugate transpose operation of a matrix respectively. Term
$\mbox{\boldmath $\mathbf{I}$}_N$ is an $N{\times}N$ identity
matrix. $\|\mathbf{a}\|$ stands for the Euclidean norm of a vector $\mathbf{a}$, and  $\overset{w.p.}{\longrightarrow}$ represents convergence with probability one. Finally, we denote the expectation operation by $\mathrm{E}\left[\cdot\right]$.

\section{System Model}
The considered MIMO multi-relay network consists of a single source and
destination node both equipped with $M$ antennas, and $K$
$N$-antenna relay nodes distributed between the source-destination
pair as illustrated in Fig.~1. When the source node implements
spatial multiplexing, the requirement $N \ge M$ must be
satisfied if each relay node is supposed to support all the $M$
independent data streams. We consider half-duplex non-regenerative
relaying throughout this paper, where it takes two non-overlapping
time slots for the data to be transmitted from the source to the
destination node via the source-to-relay ($\mathcal{S}$-$\mathcal{R}$) and relay-to-destination
($\mathcal{R}$-$\mathcal{D}$) channels. Due to deep large-scale fading effects produced by the long
distance, we assume that there is no direct link between the source
and destination. In this paper, perfect CSI of $\mathcal{S}$-$\mathcal{R}$ and imperfect CSI
of $\mathcal{R}$-$\mathcal{D}$ are assumed to be available at relay nodes. In a practical
system, each relay needs to transmit training sequences or pilots to
acquire the CSI of all the forward channels. Here we assume that the
destination node can estimate the CSI of all the forward channels
during the pilot phase. But due to large number of relay nodes, a
feedback delay and training overhead is expected at each relay. So perfect CSIs of $\mathcal{R}$-$\mathcal{D}$ are hard to be obtained at relays. We
assume that CSI of $\mathcal{R}$-$\mathcal{D}$ is imperfect with a Gaussian distributed error
at relay nodes as in~\cite{27}, and the destination
node only knows the statistical distribution of these $\mathcal{R}$-$\mathcal{D}$ CSI errors.

In the first time slot, the source node broadcasts the signal to all
the relay nodes through $\mathcal{S}$-$\mathcal{R}$ channels. Let $M{\times}1$ vector $\mbox{\boldmath
$\mathbf{s}$}$ be the transmit signal vector satisfying the power
constraint ${\rm E}\left\{ {{\bf{ss}}^H } \right\} = \left( {{P
\mathord{\left/
 {\vphantom {P M}} \right.
 \kern-\nulldelimiterspace} M}} \right){\bf{I}}_M$, where $P$ is defined
as the transmit power at the source node. Let ${\bf{H}}_k  \in
\mathbb{C}^{N \times M}$, $(k=1,...,K)$ stand for the $\mathcal{S}$-$\mathcal{R}$ MIMO
channel matrix from the source node to the $k$-th relay node. All the
relay nodes are supposed to be located in a cluster. Then all the
$\mathcal{S}$-$\mathcal{R}$ channels ${\bf{H}}_1,\cdots,{\bf{H}}_K$ can be supposed to
be independently and identically distributed ($i.i.d.$) and
experience the same Rayleigh flat fading. Assume that the entries of
${\bf{H}}_k$ are zero-mean complex Gaussian random variables with
variance one. Then the corresponding received signal at the $k$-th
relay can be written as
\begin{equation}
{\bf{r}}_k  = {\bf{H}}_k {\bf{s}} + {\bf{n}}_k,
\end{equation}
where the term $\mbox{\boldmath$\mathbf{n}$}_{k}$ is the
spatio-temporally white zero-mean complex additive Gaussian noise
vector, independent across $k$, with the covariance matrix ${\rm
E}\left\{ {{\bf{n}}_k {\bf{n}}_k^H } \right\} = \sigma _1^2
{\bf{I}}_N$. Therefore, noise variance $\sigma _1^2$ represents the
noise power at each relay node.

In the second time slot, firstly each relay node performs linear
processing by multiplying ${\bf{r}}_k$ with an $N \times N$ beamforming
matrix ${\bf{F}}_k$. This ${\bf{F}}_k$ is based on its perfect $\mathcal{S}$-$\mathcal{R}$ CSI ${\bf{H}}_k$ and  imperfect $\mathcal{R}$-$\mathcal{D}$ CSI ${\widehat{\bf{G}}}_k$.
Consequently, the signal vector sent from the $k$-th relay node is
\begin{equation}
{\bf{t}}_k  = {\bf{F}}_k {\bf{r}}_k.
\end{equation}
From more practical consideration, we assume that each relay node
has its own power constraint satisfying ${\rm E}\left\{
{{\bf{t}}_k^H {\bf{t}}_k } \right\} \leq Q$, which is independent of
power $P$. Hence a power constraint condition of ${\bf{t}}_k$ can be
derived as
\begin{equation}
p\left( {{\bf{t}}_k } \right) = \mathrm{tr}\left\{ {{\bf{F}}_k \left(
{\frac{P}{M}{\bf{H}}_k {\bf{H}}_k^H  + \sigma _1^2 {\bf{I}}_N }
\right){\bf{F}}_k^H } \right\} \le Q.
\end{equation}
After linear relay beamforming processing, all the relay nodes forward their data simultaneously to the destination.
Thus the signal vector received by the destination can be expressed as
\begin{align}\label{rzf5}
 {\bf{y}} & = \sum_{k=1}^{K}{\bf{G}}_k{\bf{t}}_k
+ {\bf{n}}_d  =  \sum_{k=1}^{K} {\bf{G}}_k {\bf{F}}_k  {\bf{H}}_k
{\bf{s}} +\sum_{k=1}^{K} {\bf{G}}_k {\bf{F}}_k {\bf{n}}_k
+{\bf{n}}_d,
\end{align}
where ${\bf{G}}_k$, under the same assumption as ${\bf{H}}_k$, is
the $M \times N$ $\mathcal{R}$-$\mathcal{D}$ channel matrix between the $k$-th relay node
and the destination. ${\bf{n}}_d \in \mathbb{C}^M$, satisfying ${\rm
E}\left\{ {{\bf{n}}_d {\bf{n}}_d^H } \right\} = \sigma _2^2
{\bf{I}}_M$, denotes the zero-mean white circularly symmetric
complex additive Gaussian noise vector at the destination node with the
noise power $\sigma _2^2$.

\section{Relay Beamforming and QR Detection}
In this section we will consider matched-filter (MF),  zero-forcing
(ZF) and regularized zero-forcing (RZF) beamforming at  relays. The
destination applied QRD detection to successively cancel the
interference from other antennas.

\subsection{Beamforming at Relay Nodes}
Denote $\widehat{\bf{G}}_k$ as the imperfect CSI of $\mathcal{R}$-$\mathcal{D}$ at the $k$-th
relay. When MF is chosen, beamforming at the $k$-th relay is
\begin{equation}
{\bf{F}}^{MF}_k  =  \widehat{\bf{G}}_k^H {\bf{H}}_k^H,
\end{equation}
where we set MF as both the receiver of $\mathcal{S}$-$\mathcal{R}$ channel and the precoder of $\mathcal{R}$-$\mathcal{D}$ channel.

When MF-ZF is chosen, beamforming at the $k$-th relay is
\begin{equation}
{\bf{F}}_k^{MF - ZF}  =  \widehat{\bf{G}}_k^H \left( {\widehat{\bf{G}}_k \widehat{\bf{G}}_k^H } \right)^{ - 1} {\bf{H}}_k^H,
\end{equation}
where we set ZF as the precoder of $\mathcal{R}$-$\mathcal{D}$ channel. Here, the requirement that $N \ge M$ is also indispensable for the inversion of $\left( {\widehat{\bf{G}}_k} \widehat{\bf{G}}_k^H  \right)$.

When MF-RZF is chosen, beamforming at the $k$-th relay is
\begin{equation}
{\bf{F}}_k^{MF - RZF}  =  \widehat{\bf{G}}_k^H \left( {\widehat{\bf{G}}_k \widehat{\bf{G}}_k^H  + \alpha_k {\bf{I}}_M } \right)^{ - 1} {\bf{H}}_k^H,
\end{equation}
where we set RZF as the precoder of $\mathcal{R}$-$\mathcal{D}$ channel. Note that MF-ZF is a special case of MF-RZF when $\alpha_k=0$.

\subsection{QR Decomposition and SIC Detection}
QR-decomposition (QRD) detector is utilized
as the destination receiver $\bf{W}$ in this paper, which is proved
to be asymptotically equivalent to that of the maximum-likelihood detector (MLD)~\cite{13}.
Let $\sum_{k=1}^{K}{\bf{G}}_k {\bf{F}}_k {\bf{H}}_k={\bf{H}}_{\mathcal {S}\mathcal {D}} $. Then (4) can be rewritten as
\begin{equation}
 {\bf{y}}= {\bf{H}}_{\mathcal {S}\mathcal {D}} {\bf{s}}
+ \widehat{\bf{n}},
\end{equation}
where ${\bf{H}}_{\mathcal {S}\mathcal {D}}$ represents the effective channel between the source
and destination node, and $\widehat{\bf{n}}=\sum_{k=1}^K {\bf{G}}_k {\bf{F}}_k {\bf{n}}_k
+{\bf{n}}_d$ is the effective noise vector cumulated from
the noise $\bf{n}_k$ at the $k$-th relay node, and
 the noise vector ${\bf{n}}_d$ at the destination. Finally, in order to cancel the interference from other antennas,
QR decomposition of the effective channel is implemented as
\begin{equation}\label{5}
  {\bf{H}}_{\mathcal {S}\mathcal {D}}={\bf{Q}}_{\mathcal {S}\mathcal {D}} {\bf{R}}_{\mathcal {S}\mathcal {D}},
\end{equation}
where ${\bf{Q}}_{\mathcal {S}\mathcal {D}}$ is an $M \times M$
unitary matrix and ${\bf{R}}_{\mathcal {S}\mathcal {D}}$ is an $M
\times M$ right upper triangular matrix. Therefore the QRD  detector
at destination node is chosen as: ${\bf{W}}={\bf{Q}}_{\mathcal
{S}\mathcal {D}}^H $, and the signal vector after QRD detection becomes
\begin{equation}
 \tilde {\bf{y}}={\bf{Q}}_{\mathcal {S}\mathcal {D}}^H {\bf{y}}= {\bf{R}}_{\mathcal {S}\mathcal {D}} {\bf{s}}
+ {\bf{Q}}_{\mathcal {S}\mathcal {D}}^H {\bf{n}}.
\end{equation}

A power control factor $\widehat{\rho_k} $ is set with $\bf{F}_k$ in (2) to
guarantee that the $k$-th relay transmit power is equal to $Q$. The transmit signal from each
relay node after linear beamforming and power control becomes
\begin{equation}
 {\bf{t}}_k  = \widehat{\rho_k} {\bf{F}}_k {\bf{r}}_k,
\end{equation}
where the power control factor $\widehat{\rho_k}$ can be derived from (3) as
\begin{equation}
 \widehat{\rho _k}  = \biggl ( Q \biggl / \mathrm{tr} \biggl \{ {\bf{F}}_k \left( {\frac{P}{M}{\bf{H}}_k {\bf{H}}_k^H  + \sigma _1^2 {\bf{I}}_N } \right){\bf{F}}_k^H \biggl \} \biggl ) ^{\frac{1}{2}}.
\end{equation}

From the cut-set theorem in network information theory [4], the upper bound
capacity of the dual-hop MIMO relay networks is
\begin{equation}
C_{upper}  =  \mathrm{E}_{\left\{ {{\bf{H}}_k } \right\}_{k = 1}^K } \left\{ {\frac{1}{2}\log \det \left( {{\bf{I}}_M  + \frac{P}{{M \sigma _1^2 }}\sum\limits_{k = 1}^K {{\bf{H}}_k^H {\bf{H}}_k } } \right)} \right\}.
\end{equation}

\section{ Ergodic Capacities}
In this section, we will derive the ergodic capacities of the three relay
beamformers under the condition of imperfect $\mathcal{R}$-$\mathcal{D}$ CSI at relays. These ergodic capacities are based on the instantaneous SNR of source-to-destination channel.
For the $k$-th relay, we denote the accurate $\mathcal{R}$-$\mathcal{D}$ channel matrix as ${\bf{G}}_k$. We assume
the $\mathcal{R}$-$\mathcal{D}$ CSI error caused by large delay and reciprocity mismatch to be complex Gaussian distributed, and model the imperfect $\mathcal{R}$-$\mathcal{D}$ CSI received at the $k$-th relay as~\cite{27},
\begin{equation}
{\widehat{\bf{G}}}_k={\bf{G}}_k+e{\bf{\Omega}}_{k},
\end{equation}
where ${\bf{\Omega}}_{k}$ is a matrix independent of ${\bf{G}}_k$, whose entries are $i.i.d$ zero-mean complex Gaussian, with unity variance, and $e$ is the gain of channel loss. Therefore the power of CSI error is $e^{2}$~\cite{14}. In this paper, we consider $e\ll 1$.

\subsection{Preliminaries}
In this subsection, we give some lemmas associated with ${\bf{\Omega}}_k$, which will be used to derive the instantaneous SNR of source-to-destination channel.
\newtheorem{theorem}{Theorem}
\newtheorem{lemma}[theorem]{Lemma}
\begin{lemma}
$\mathrm{E} \left [ \mathrm{tr}\left( {\bf{A}}{\bf{\Omega}}_k^H \right){\bf{\Omega}}_k\right ]={\bf{A}}$
for any complex matrix ${\bf{A}}\in\mathbb{C}^{M\times N}$.
\end{lemma}
\begin{proof}
Note that $\mathrm{E}\left[ ({\bf{\Omega}}_k)_{i,j} \right]=0$ and $\mathrm{E}\left[ ({\bf{\Omega}}_k)_{i,j} ({\bf{\Omega}}_k)_{i,j}^{*} \right]=1$. Since the entries in ${\bf{\Omega}}_k$ are independent, we have
\begin{multline}
\mathrm{E}\left[\left(\mathrm{tr}\left( {\bf{A}} {\bf{\Omega}}_k^H\right){\bf{\Omega}}_k\right)_{i,j}\right]\\=\mathrm{E}\left[\Sigma_{m=1}^M \Sigma_{l=1}^N \left({\bf{A}}\right)_{m,l}\left({\bf{\Omega}}_k\right)_{m,l}^{*}\left({\bf{\Omega}}_k\right)_{i,j}\right]
\\=\left({\bf{A}}\right)_{i,j}.
\end{multline}
\end{proof}
\begin{lemma}
$\mathrm{E} \left [ \mathrm{tr}\left( {\bf{A}}{\bf{\Omega}}_k^H \right)\mathrm{tr}\left( {\bf{B}}{\bf{\Omega}}_k \right)\right ]=\mathrm{tr}\left(\mathbf{AB}\right)$, for any $\mathbf{A}\in\mathbb{C}^{M\times N}$ and $\mathbf{B}\in\mathbb{C}^{N\times M}$.
\end{lemma}
\begin{proof}
\begin{multline}
\mathrm{E} \left [ \mathrm{tr}\left( {\bf{A}}{\bf{\Omega}}_k^H \right)\mathrm{tr}\left( {\bf{B}}{\bf{\Omega}}_k \right)\right ]\\=\mathrm{E} \left [ \left(\Sigma_{m=1}^M \Sigma_{n=1}^N \left({\bf{A}}\right)_{m,n}\left({\bf{\Omega}}_k\right)_{m,n}^*\right)\right.\\\left.\left(\Sigma_{n=1}^N \Sigma_{m=1}^M \left({\bf{B}}\right)_{n,m}\left({\bf{\Omega}}_k\right)_{m,n}\right)\right ]\\
=\Sigma_{m=1}^M \Sigma_{n=1}^N \left({\bf{A}}\right)_{m,n}\left({\bf{B}}\right)_{n,m}=\mathrm{tr}\left(\mathbf{AB}\right).
\end{multline}
\end{proof}
\begin{lemma}
$\mathrm{E} \left [ {\bf{\Omega}}_k {\bf{A}}{\bf{\Omega}}_k^H \right ]=\mathrm{tr}\left({\bf{A}}\right){\bf{I}}_M$
 for any $\mathbf{A}\in\mathbb{C}^{M\times M}$.
\end{lemma}
\begin{proof}
\begin{multline}
\mathrm{E}\left[({\bf{\Omega}}_k {\bf{A}} {\bf{\Omega}}_k^H)_{i,j}\right]=\sum_{m=1}^M \sum_{n=1}^M \mathrm{E}\left[({\bf{\Omega}}_k)_{i,m} ({\bf{A}})_{m,n} ({\bf{\Omega}}_k^H)_{n,j}\right]\\=\sum_{m=1}^M \sum_{n=1}^M ({\bf{A}})_{m,n} \mathrm{E}\left[({\bf{\Omega}}_k)_{i,m}  ({\bf{\Omega}}_k)_{j,n}^{*} \right]=\mathrm{tr}\left(\mathbf{A}\right)\delta\left[m-n\right],
\end{multline}
where $\delta\left[x\right]=1$ for $x=0$ and $0$ otherwise.
\end{proof}
\subsection{Instantaneous SNR of MF Beamforming}

In the presence of imperfect $\mathcal{R}$-$\mathcal{D}$ CSI at relays, the MF beamforming at the
$k$-th relay is
\begin{equation}
{{\bf{F}}}_{k}={\widehat{\bf{G}}}_k^H {\bf{H}}_k^H.
\label{eq1}
\end{equation}
Thus, the signal vector received by the destination is
\begin{equation}
{\bf{y}}=\sum_{k = 1}^K \widehat{\rho}_k {\bf{G}}_k {\widehat{\bf{G}}}_k^H {\bf{H}}_k^H \left( {\bf{H}}_k {\bf{s}} + {\bf{n}}_k \right) + {\bf{n}}_{d}.
\label{eq2}
\end{equation}
Since $e\ll 1$ and the $k$-th relay node only knows an imperfect
$\mathcal{R}$-$\mathcal{D}$ CSI ${\widehat{\bf{G}}}_k$, the power
control factor becomes
\begin{equation}
\begin{split}
\widehat{\rho _k}  &= \left ( Q \biggl / \mathrm{tr} \left \{ {\widehat{\bf{G}}}_k^H {\bf{H}}_k^H \left( {\frac{P}{M}{\bf{H}}_k {\bf{H}}_k^H  + \sigma _1^2 {\bf{I}}_N } \right){\bf{H}}_k {\widehat{\bf{G}}}_k \right \} \right ) ^{\frac{1}{2}}
\\&\cong \rho_k^{\frac{1}{2}}\left(1-\frac{e{v}_k}{2{u}_k}\right),
\end{split}
\end{equation}
where
\begin{equation}
{u}_k=\mathrm{tr}\left ({{\bf{G}}}_k^H {\bf{H}}_k^H \left( {\frac{P}{M}{\bf{H}}_k {\bf{H}}_k^H  + \sigma _1^2 {\bf{I}}_N } \right){\bf{H}}_k {{\bf{G}}}_k\right ),
\end{equation}
\begin{equation}
v_k=\mathrm{tr}\left ( {\bf{H}}_k^H \left( {\frac{P}{M}{\bf{H}}_k {\bf{H}}_k^H  + \sigma _1^2 {\bf{I}}_N } \right){\bf{H}}_k \left ({{\bf{G}}}_k {\bf{\Omega}}_k^H +{\bf{\Omega}}_k {{\bf{G}}}_k^H\right )\right ),
\end{equation}
and
\begin{equation}
\rho_k=\left(\frac{Q}{{u}_k}\right)^{\frac{1}{2}}.
\end{equation}
Using Lemma~1 and Lemma~2, through some manipulations and omitting
some relatively small terms, (\ref{eq2}) becomes
\begin{equation} \label{eq3}
\begin{split}
{\bf{y}}&\cong\sum_{k = 1}^K {\rho}_k {\bf{G}}_k {\bf{G}}_k^H {\bf{H}}_k^H {\bf{H}}_k {\bf{s}}
 + e \sum_{k = 1}^K {\rho}_k {\bf{G}}_k  {\bf{\Omega}}_k^H {\bf{H}}_k^H {\bf{H}}_k {\bf{s}}
 \\ &+ \sum_{k = 1}^K {\rho}_k {\bf{G}}_k {\bf{G}}_k^H {\bf{H}}_k^H {\bf{n}}_k
 + {\bf{n}}_d.
\end{split}
\end{equation}
We observe that the second term in the right-hand side of (\ref{eq3}) is the additional noises introduced by the $\mathcal{R}$-$\mathcal{D}$ CSI error. The third term is the noise introduced by the noises at each relays.
We denote the last three terms in the right-hand side of (\ref{eq3}) as
\begin{equation}
{\widehat{\bf{n}}}=e \sum_{k = 1}^K \rho_k {\bf{G}}_k  {\bf{\Omega}}_k^H {\bf{H}}_k^H {\bf{H}}_k {\bf{s}}  + \sum_{k = 1}^K \rho_k {\bf{G}}_k {\bf{G}}_k^H {\bf{H}}_k^H {\bf{n}}_k + {\bf{n}}_d,
\end{equation}
 and refer to it as the effective post-processing noise. Denote
\begin{equation}
{\bf{H}}_{{\mathcal {S}\mathcal {D}},MF}=\sum_{k=1}^K \rho_k {\bf{G}}_k {\bf{G}}_k^H {\bf{H}}_k^H {\bf{H}}_k,
\end{equation}
as the effective transmitting matrix of the whole network.
Then (\ref{eq3}) becomes
\begin{equation}
{\bf{y}}={\bf{H}}_{{\mathcal {S}\mathcal {D}},MF} \mathbf{s}+\widehat{\bf{n}}.
\end{equation}
Using QRD ${\bf{H}}_{{\mathcal {S}\mathcal {D}},MF}=\mathbf{Q}_{MF} \mathbf{R}_{MF}$ at the destination,
we have
\begin{equation}
\mathbf{Q}_{MF}^H{\bf{y}}=\mathbf{R}_{MF}\mathbf{s}+\mathbf{Q}_{MF}^H\widehat{\bf{n}}.
\end{equation}
So the power of the $m$th transmitted signal stream becomes $P/M {(\mathbf{R}_{MF})}_{m,m}^2$.
We now calculate the covariance matrix of the effective post-processing noise $\widehat{\bf{n}}$ as,
\begin{equation}
\begin{split}
&\mathrm{E} \biggl [\mathbf{Q}_{MF}^H{\widehat{\bf{n}}} (\mathbf{Q}_{MF}^H{\widehat{\bf{n}}})^H \biggl ]\\
&=  e^2 \frac{P}{M} \sum_{k=1}^K {\rho_k}^2 \mathbf{Q}_{MF}^H{\bf{G}}_k \mathrm{tr} \left(   ({\bf{H}}_k^H {\bf{H}}_k)^2\right  ) {\bf{G}}_k^H \mathbf{Q}_{MF}\\&+ \sigma_1^2 \sum_{k=1}^K \rho_k^2 \mathbf{Q}_{MF}^H{\bf{G}}_k {\bf{G}}_k^H {\bf{H}}_k^H {\bf{H}}_k {\bf{G}}_k {\bf{G}}_k^H \mathbf{Q}_{MF}+ \sigma_2^2 \mathbf{I}_M,
\end{split}
\end{equation}
where we used Lemma~3.
So the effective noise power of the $m$th data stream is
\begin{equation}
\begin{split}
&\mathrm{E}\biggl[(\widehat{\bf{n}}\widehat{\bf{n}}^H)_{m,m}\biggl]\\&=
e^2 \frac{P}{M} \sum_{k=1}^K {\rho_k}^2  \mathrm{tr} \left(   ({\bf{H}}_k^H {\bf{H}}_k)^2\right  ) \|\left(\mathbf{Q}_{MF}^H {\bf{G}}_k\right)_m \|^2\\&+ \sigma_1^2 \sum_{k=1}^K \rho_k^2 \|\left(\mathbf{Q}_{MF}^H{\bf{G}}_k {\bf{G}}_k^H {\bf{H}}_k^H\right)_m \|^2+ \sigma_2^2.
\end{split}
\end{equation}
Thus, the post-processing SNR per symbol of the $m$th stream can be expressed as (\ref{eq4}) at the top of next page.
\begin{figure*}[!t]
\begin{equation}
\gamma_m^{MF}=\frac {\frac{P}{M} {(\mathbf{R}_{MF})}_{m,m}^2} {\underbrace{e^2 \frac{P}{M} \sum_{k=1}^K {\rho_k}^2  \mathrm{tr} \left(   ({\bf{H}}_k^H {\bf{H}}_k)^2\right  ) \|\left(\mathbf{Q}_{MF}^H {\bf{G}}_k \right)_m\|^2}_{channel-error-generated~ noise~power} +
\sigma_1^2 \sum_{k=1}^K \rho_k^2 \|\left(\mathbf{Q}_{MF}^H{\bf{G}}_k {\bf{G}}_k^H {\bf{H}}_k^H\right)_m \|^2+ \sigma_2^2}
\label{eq4}
\end{equation}
\end{figure*}

Compared to the covariance of the effective noise under the
condition of  perfect CSI ($e=0$), we see that the covariance of the
effective post-processing noise under the condition of imperfect
$\mathcal{R}$-$\mathcal{D}$ CSI consists of an additional term,
which is related to transmit power $P$, $\mathcal{R}$-$\mathcal{D}$
CSI error gain $e$, and the CSIs of $\mathcal{S}-\mathcal{R}$ and
$\mathcal{R}-\mathcal{D}$. We call this term as
channel-error-generated noise power (CEG-noise power).


\subsection{Instantaneous SNR of MF-ZF Beamforming}
In the presence of imperfect CSI of $\mathcal{R}$-$\mathcal{D}$, the ZF beamforming at the
$k$-th relay is
\begin{equation}
{\widehat{\bf{F}}}_k=\widehat{\bf{G}}_k^{\dag} {\bf{H}}_k^H,
\end{equation}
where $\widehat{\bf{G}}_k^{\dag}=\widehat{\bf{G}}_k^H {(\widehat{\bf{G}}_k \widehat{\bf{G}}_k^H)}^{-1}$ is the pseudo-inverse of the matrix $\widehat{\bf{G}}_k$.
The signal vector received by the destination is
\begin{equation}
{\bf{y}}=\sum_{k = 1}^K \widehat{\rho_k} {\bf{G}}_k {\widehat{\bf{G}}}_k^{\dag} {\bf{H}}_k^H \left( {\bf{H}} {\bf{s}} + {\bf{n}}_k \right) + {\bf{n}}_{d},
\label{eq6}
\end{equation}
Since $e \ll 1$, the pseudo-inverse of matrix ${\bf{G}}_k$ can be approximated using the Taylor expansion as
\begin{equation}\label{eq5}
\widehat{\bf{G}}_k^{\dag} \cong ({\bf{I}}_N-e {\bf{G}}_k^{\dag}
{\bf{\Omega}}_k) {\bf{G}}_k^{\dag},
\end{equation}
where ${\bf{G}}_k^{\dag}={\bf{G}}_k^{H}({\bf{G}}_k{\bf{G}}_k^{H})^{-1}$ is the pseudo-inverse of ${\bf{G}}_k$. The relay power control factor can be approximated as
\begin{equation}\label{6}
\begin{split}
&\widehat{\rho _k}  = \left ( Q \biggl / \mathrm{tr} \left \{ {\widehat{\bf{G}}}_k^{\dag} {\bf{H}}_k^H \left( {\frac{P}{M}{\bf{H}}_k {\bf{H}}_k^H  + \sigma _1^2 {\bf{I}}_N } \right){\bf{H}}_k \left({\widehat{\bf{G}}}_k^{\dag}\right)^H \right \} \right ) ^{\frac{1}{2}}
\\&\cong \rho_k^{\frac{1}{2}}\left(1+\frac{e{v}_k}{2{u}_k}\right),
\end{split}
\end{equation}
where
\begin{equation}
{u}_k=\mathrm{tr}\left ({\bf{H}}_k^H \left( {\frac{P}{M}{\bf{H}}_k {\bf{H}}_k^H  + \sigma _1^2 {\bf{I}}_N } \right){\bf{H}}_k \left({\bf{G}}_k{\bf{G}}_k^H\right)^{-1}\right ),
\end{equation}
\begin{multline}
v_k=\mathrm{tr}\biggl ( {\bf{H}}_k^H \left( {\frac{P}{M}{\bf{H}}_k {\bf{H}}_k^H  + \sigma _1^2 {\bf{I}}_N } \right){\bf{H}}_k \left ({\bf{G}}_k{\bf{G}}_k^H\right)^{-1}\\\left( {\bf{G}}_k{\bf{\Omega}}_k^H +{\bf{\Omega}}_k {\bf{G}}_k^H\right)\left({\bf{G}}_k{{\bf{G}}}_k^H\right )^{-1}\biggl ),
\end{multline}
and
\begin{equation}
\rho_k=\left(\frac{Q}{{u}_k}\right)^{\frac{1}{2}}.
\end{equation}
By substituting (\ref{eq5}) and (\ref{6}) into (\ref{eq6}) and omitting some relatively small terms, the received signal at the destination can be further written as
\begin{equation}
\begin{split}\label{eq7}
{\bf{y}}&\cong \sum_{k=1}^K \rho_k {{\bf{H}}_k}^H {{\bf{H}}_k} {\bf{s}}+ \sum_{k=1}^K \rho_k {{\bf{H}}_k}^H {\bf{n}}_k\\
&-e \sum_{k=1}^K \rho_k {\bf{\Omega}}_k {{\bf{G}}_k}^{\dag} {{\bf{H}}_k}^H {{\bf{H}}_k} {\bf{s}} +{\bf{n}}_{d},
\end{split}
\end{equation}
where we used Lemma~1 and Lemma~2.
Just as the case of MF beamforming, the third
term in the right-hand side of (\ref{eq7})  are caused by the $\mathcal{R}$-$\mathcal{D}$
CSI error. The effective post-processing noise is the last three
terms in the right-hand side of (\ref{eq7}), i.e.,
\begin{equation}
\widehat{\bf{n}} = \sum_{k=1}^K \rho_k {{\bf{H}}_k}^H {\bf{n}}_k
-e \sum_{k=1}^K \rho_k {\bf{\Omega}}_k {{\bf{G}}_k}^{\dag} {{\bf{H}}_k}^H {{\bf{H}}_k} {\bf{s}}
+{\bf{n}}_{d}.
\end{equation}
Denote
\begin{equation}
{\bf{H}}_{{\mathcal {S}\mathcal {D}},MF-ZF}=\sum_{k=1}^K \rho_k {\bf{H}}_k^H {\bf{H}}_k.
\end{equation}
as the effective transmitting matrix. Using QRD as ${\bf{H}}_{{\mathcal {S}\mathcal {D}},MF-ZF}={\bf{Q}}_{MF-ZF}{\bf{R}}_{MF-ZF}$, we have
\begin{equation}
{\bf{y}}={\bf{H}}_{{\mathcal {S}\mathcal {D}},MF-ZF} {\bf{x}}+\widehat{\bf{n}}=
{\bf{Q}}_{MF-ZF}{\bf{R}}_{MF-ZF}\bf{x}+\widehat{\bf{n}}.
\end{equation}
The covariance matrix of the effective post-processing noise is
\begin{equation}
\begin{split}
&\mathrm{E} \biggl [\mathbf{Q}_{MF-ZF}^H{\widehat{\bf{n}}} (\mathbf{Q}_{MF-ZF}^H{\widehat{\bf{n}}})^H \biggl ]\\
&=e^2 {\sigma}_1^2 \sum_{k=1}^K \rho_k^2 \mathrm{tr} \left ( \left( {{\bf{H}}_k}^H {{\bf{H}}_k}\right)^2 \left({\bf{G}}_k {\bf{G}}_k^H\right)^{-1} \right ){\bf{I}}_M\\&+
{\sigma}_1^2 \sum_{k=1}^K \rho_k^2 \mathbf{Q}_{MF-ZF}^H{{\bf{H}}_k}^H {{\bf{H}}_k}\mathbf{Q}_{MF-ZF} + {\sigma}_2^2 {\bf{I}}_M,
\label{eq15}
\end{split}
\end{equation}
where we used Lemma~3 in the derivation.  The
above formulas result in the post-processing SNR per symbol of the
$m$th stream as (\ref{eq10}) at the top of next page.
\begin{figure*}[!t]
\begin{equation}
\gamma_m^{MF-ZF}=\frac {\frac{P}{M} (\mathbf{R}_{MF-ZF})_{m,m}^2} {\underbrace{e^2 \frac{P}{M} \sum_{k=1}^K \rho_k^2 \mathrm{tr} \left ( ( {{\bf{H}}_k}^H {{\bf{H}}_k})^2 ({\bf{G}}_k {\bf{G}}_k^H)^{-1} \right )}_{channel-error-generated~noise~power} + \sigma_1^2 \sum_{k=1}^K
\rho_k^2 \|\left(\mathbf{Q}_{MF-ZF}^H{\bf{H}}_k^H \right)_m\|^2 + \sigma_2^2}.
\label{eq10}
\end{equation}
\hrulefill
\end{figure*}

Once again, the imperfect $\mathcal{R}$-$\mathcal{D}$ CSI generates the CEG-noise power term in the post-processing SNR.
\subsection{Instantaneous SNR of MF-RZF Beamforming}
To simplify the analysis, consider $\alpha_k=\alpha$.
In the presence of imperfect CSI of $\mathcal{R}$-$\mathcal{D}$ at each relay, the MF-RZF
beamforming at the $k$-th relay is
\begin{equation}\label{rzf1}
{\widehat{\bf{F}}}_k=\widehat{{\bf{G}}}_k^H (\widehat{{\bf{G}}}_k \widehat{{\bf{G}}}_k^H
+ \alpha \textbf{I}_M)^{-1} {\bf{H}}_k^H.
\end{equation}
When $e \ll 1$, using Taylor expansion, we have
\begin{equation}\label{rzf2}
\begin{split}
&\widehat{{\bf{G}}}_k^H (\widehat{{\bf{G}}}_k \widehat{{\bf{G}}}_k^H
+ \alpha \textbf{I}_M)^{-1} \\&\cong ({\bf{G}}_k^{\dag} -e
{\bf{G}}_k^{\dag}{\bf{\Omega}}_k{\bf{G}}_k^{\dag})(\textbf{I}_M-\alpha
{\bf{G}}_k^{\alpha}+e\alpha {\bf{G}}_k^{\alpha} {\bf{G}}_k^e
{\bf{G}}_k^{\alpha})\\&={\bf{G}}_k^H ({\bf{G}}_k {\bf{G}}_k^H
+ \alpha \textbf{I}_M)^{-1}-e {\bf{G}}_k^E,
\end{split}
\end{equation}
where
\begin{equation}\label{rzf3}
{\bf{G}}_k^{\alpha}=({\bf{G}}_k {\bf{G}}_k^H + \alpha
\textbf{I}_M)^{-1}
\quad
{\bf{G}}_k^E={\bf{G}}_k^{\dag}
{\bf{\Omega}}_k {\bf{G}}_k^{H} \mathbf{G}_k^\alpha.
\end{equation}
If $\alpha$ increases,
the norm of entries of ${\bf{G}}_k^E$ will decrease. So the impact
of imperfect CSI will be reduced. In addition, the power control factor $\widehat{\rho_k}$ will increase, resulting in a reduced unnormalized transmit power at relays.
But if
$\alpha$ becomes too large, interference from other antennas at
the destination will be considerable. In this paper, we try to  obtain the optimal $\alpha$ to maximize the SINR on each data stream.
Let ${\bf{G}}_k^e={\bf{G}}_k {\bf{\Omega}}_k^H+{\bf{\Omega}}_k {\bf{G}}_k^H$.
The relay power control factor can be approximated as
\begin{equation}\label{rzf4}
\widehat{\rho _k}
\cong \rho_k^{\frac{1}{2}}\left(1+\frac{e{v}_k}{2{u}_k}\right),
\end{equation}
where
\begin{equation}
{u}_k=\mathrm{tr}\left ({\bf{H}}_k^H \left( {\frac{P}{M}{\bf{H}}_k {\bf{H}}_k^H  + \sigma _1^2 {\bf{I}}_N } \right){\bf{H}}_k \left({\bf{G}}_k^\alpha-\alpha\left({\bf{G}}_k^\alpha\right)^2\right)\right ),
\end{equation}
\begin{multline}
v_k=\mathrm{tr}\biggl ( {\bf{H}}_k^H \left( {\frac{P}{M}{\bf{H}}_k {\bf{H}}_k^H  + \sigma _1^2 {\bf{I}}_N } \right){\bf{H}}_k {\bf{G}}_k^\alpha\\\left ({\bf{G}}_k^e{\bf{G}}_k^\alpha{\bf{G}}_k{\bf{G}}_k^H+{\bf{G}}_k{\bf{G}}_k^H{\bf{G}}_k^\alpha{\bf{G}}_k^e-{\bf{G}}_k^e\right){\bf{G}}_k^\alpha\biggl),
\end{multline}
and
\begin{equation}
\rho_k=\left(\frac{Q}{{u}_k}\right)^{\frac{1}{2}}.
\end{equation}
Using Lemma~1 and Lemma~2, substituting (\ref{rzf1}), (\ref{rzf2}), (\ref{rzf3}) and (\ref{rzf4}) into (\ref{rzf5}) and omitting some relatively small terms, the received signal at destination can be expanded as
\begin{equation}\label{eq12}
\begin{split}
&{\bf{y}}
=\sum_{k=1}^K \rho_k (\textbf{I}_M- \alpha{\bf{G}}_k^{\alpha}){{\bf{H}}_k}^H {{\bf{H}}_k} {\bf{s}}
\\&+ \sum_{k=1}^K \rho_k (\textbf{I}_M-\alpha {\bf{G}}_k^{\alpha}){{\bf{H}}_k}^H {\bf{n}}_k \\&
+e \sum_{k=1}^K \rho_k {\bf{\Omega}}_k
{\bf{G}}_k^{H} \mathbf{G}_k^\alpha {{\bf{H}}_k}^H {{\bf{H}}_k} {\bf{s}}
+{\bf{n}}_{d}.
\end{split}
\end{equation}
Denote
\begin{equation}
{\bf{H}}_{{\mathcal {S}\mathcal {D}},MF-RZF}=\sum_{k=1}^K \rho_k
(\textbf{I}_M- \alpha{\bf{G}}_k^{\alpha}){{\bf{H}}_k}^H {{\bf{H}}_k}.
\end{equation}
Then (\ref{eq12}) becomes
\begin{equation}
{\bf{y}}={\bf{H}}_{{\mathcal {S}\mathcal {D}},MF-RZF} {\bf{x}}+\widehat{\bf{n}}\triangleq
{\bf{Q}}_{{\mathcal {S}\mathcal {D}},MF-RZF}{\bf{R}}_{{\mathcal {S}\mathcal {D}},MF-RZF} \bf{x}+\widehat{\bf{n}},
\end{equation}
where the effective post-processing noise
\begin{equation}
\widehat{\bf{n}}= e \sum_{k=1}^K
\rho_k {\bf{\Omega}}_k
{\bf{G}}_k^{H} \mathbf{G}_k^\alpha {{\bf{H}}_k}^H {{\bf{H}}_k}
{\bf{s}}+\sum_{k=1}^K \rho_k(\textbf{I}_M-\alpha
{\bf{G}}_k^{\alpha}){{\bf{H}}_k}^H {\bf{n}}_k +{\bf{n}}_{d}. \label{eq9}
\end{equation}
The covariance of effective post-processing
noise is
\begin{equation}
\begin{split}
&\mathrm{E} \biggl [\mathbf{Q}_{MF-RZF}^H\widehat{\bf{n}} ({\mathbf{Q}_{MF-RZF}^H\widehat{\bf{n}}})^H \biggl
]\\&=e^2 \frac{P}{M}
\sum_{k=1}^K \rho_k^2 \mathrm{tr}\left(\left({{\bf{H}}_k}^H {{\bf{H}}_k}\right)^2 \mathbf{{G}}_k^\alpha\left(\mathbf{I}_M-\alpha\mathbf{{G}}_k^\alpha\right) \right) \\&+{\sigma}_1^2 \sum_{k=1}^K \rho_k^2 \mathbf{Q}_{MF-RZF}^H\left(\textbf{I}_M-\alpha
{\bf{G}}_k^{\alpha}\right){{\bf{H}}_k}^H\\&
{{\bf{H}}_k}\left(\textbf{I}_M-\alpha{\bf{G}}_k^{\alpha}\right)^H \mathbf{Q}_{MF-RZF}+ {\sigma}_2^2
{\bf{I}}_M,
\end{split}
\end{equation}
where we used Lemma~3.
Then the post-processing SNR per
symbol of the $m$th stream is calculated as (\ref{eq11}) at the top of next page.
\begin{figure*}[ht]
\begin{equation}
\gamma_m^{MF-RZF}=\frac {\frac{P}{M} (\mathbf{R}_{MF-RZF})_{m,m}^2} {\underbrace{e^2 \frac{P}{M}
\sum_{k=1}^K \rho_k^2 \mathrm{tr}\left(({{\bf{H}}_k}^H {{\bf{H}}_k})^2\mathbf{{G}}_k^\alpha(\mathbf{I}_M-\alpha\mathbf{{G}}_k^\alpha)\right)}
_{channel-error-generated~noise~power}
+ \sigma_1^2 \sum_{k=1}^K \rho_k^2 \|\left(\mathbf{Q}_{MF-RZF}^H(\textbf{I}_M-\alpha
{\bf{G}}_k^{\alpha}){\bf{H}}_k^H\right)_m \|^2 + \sigma_2^2}. \label{eq11}
\end{equation}
\hrulefill
\end{figure*}

\subsection{Ergodic capacity}
The ergodic capacity is derived by summing up all the data rates on each antenna link, i.e.,
\begin{equation}
C = \mathrm{E}_{\left\{ {{\bf{H}}_k ,{\bf{G}}_k } \right\}_{k = 1}^K } \left\{ {\frac{1}{2}\sum\limits_{m = 1}^M {\log _2 \left( {1 + \gamma_m } \right)} } \right\}.
\end{equation}
We can see, from the instantaneous SNRs in (\ref{eq4}),
(\ref{eq10}) and (\ref{eq11}),  that a ceiling effect~\cite{27} can be
expected when PNR ($P/\sigma_1^2$) and QNR ($Q/\sigma_2^2$) $
\rightarrow \infty$. This is because that $\gamma_m$ can not tend to
infinity due to the CEG-noise power term in the denominator of
$\gamma_m$, which will also increase as the effective power does. Simulations will confirm the ceiling effect.


\section{Asymptotic capacities and the optimized MF-RZF}
To further investigate the capacity performance under imperfect CSI
of $\mathcal{R}$-$\mathcal{D}$ channel at relays, we derive
asymptotic capacities for large $K$. To simplify the analysis, we
assume a fixed power control factor for all relays in this section.
Simulation results will validate this assumption.  The power control
factor is chosen as the average one, i.e.,
\begin{equation}
\rho_k=\left(Q\biggl/\mathrm{E}\left[\mathrm{tr}\left(\mathbf{F}_k(\frac{P}{M}\mathbf{H}_k\mathbf{H}_k^H+\sigma_1^2\mathbf{I}_N)\mathbf{F}_k^H\right)\right]\right)^{\frac{1}{2}}. \end{equation}
Then
\begin{equation}
\rho_{MF}=\left(Q\biggl/\left(\left(P\left(M+N\right)+M\right)N^2\right)\right)^{\frac{1}{2}},
\end{equation}
\begin{equation}
\rho_{MF-ZF}=\left(Q\left(N-M\right)\biggl/\left(\left(P\left(M+N\right)+M\right)\right)\right)^{\frac{1}{2}},
\end{equation}
and
\begin{equation}\label{8}
\begin{split}
&\rho_{MF-RZF}\\&=\left(Q\biggl/\left(\left(P\left(M+N\right)N+MN\right)\right)\mathrm{E}\left[\frac{\lambda}{\left(\lambda+\alpha\right)^2}\right]\right)^{\frac{1}{2}}.
\end{split}
\end{equation}
In (\ref{8}), we used the decomposition $\mathbf{G}\mathbf{G}^H=\mathbf{U}\mathbf{\Lambda}\mathbf{U}^H$, where       $\mathbf{\Lambda}=\mathrm{diag}\{\lambda_1,\ldots,\lambda_M\}$ and $\mathbf{U}$ are independent to each other~\cite{25}.
For the case of large $K$, using Law of Large Number, we have
\begin{equation}
{\bf{H}}_{{\mathcal {S}\mathcal {D}},MF}\overset{w.p.}{\longrightarrow}K\left(\mathrm{E}\left[\mathbf{G}_k\mathbf{G}_k^H\mathbf{H}_k^H\mathbf{H}_k\right]\right)
=KN^2\mathbf{I}_M,
\end{equation}
\begin{equation}
{\bf{H}}_{{\mathcal {S}\mathcal {D}},MF-ZF}\overset{w.p.}{\longrightarrow}K\left(\mathrm{E}\left[\mathbf{H}_k^H\mathbf{H}_k\right]\right)=KN\mathbf{I}_M,
\end{equation}
and
\begin{multline}
{\bf{H}}_{{\mathcal {S}\mathcal {D}},MF-RZF}\overset{w.p.}{\longrightarrow}K\left(\mathrm{E}\left[(\mathbf{I}_M-\alpha\mathbf{G}_k^\alpha)\mathbf{H}_k^H\mathbf{H}_k\right]\right)
\\=KN\mathrm{E}\left[\mathbf{U}\mathrm{diag}\left\{\frac{\lambda_1}{\lambda_1+\alpha},\ldots,\frac{\lambda_M}{\lambda_M+\alpha}\right\}\mathbf{U}^H\right].
\end{multline}
Note that
\begin{multline}
\mathrm{E}\left[\left(\mathbf{U}\mathrm{diag}\left\{\frac{\lambda_1}{\lambda_1+\alpha},\ldots,\frac{\lambda_M}{\lambda_M+\alpha}\right\}\mathbf{U}^H\right)_{m,n}\right]\\
=\mathrm{E}\left[\Sigma_{j=1}^M(\mathbf{U})_{m,j}\frac{\lambda_j}{\lambda_j+\alpha}(\mathbf{U})_{n,j}^*\right]
=\mathrm{E}\left[\frac{\lambda}{\lambda+\alpha}\right]\delta\left[m-n\right].
\end{multline}
Since the asymptotic effective channel matrices are all diagonal, we have
$\mathbf{Q}\overset{w.p.}{\longrightarrow}\mathbf{I}_M$ for large $K$. Since $(\mathbf{G}_k\mathbf{G}_k^H)^{-1}$ is a complex inverse Wishart distribution with $N$ degrees of freedom~\cite{22}. We have   $\mathrm{E}\left[(\mathbf{G}_k\mathbf{G}_k^H)^{-1}\right]=\frac{M}{N-M}$~\cite{28} and   $\mathrm{E}\left[(\mathbf{H}_k^H\mathbf{H}_k)^2\right]=(MN+N^2)\mathbf{I}_M$~\cite{30}.
Using Law of Large Number, for large $K$, we have the asymptotic capacities in (\ref{7}), (\ref{2}), and (\ref{1}) at the top of next page.
\begin{figure*}[t]
\begin{multline}\label{7}
C_{MF}\overset{w.p.}{\longrightarrow}\\
\frac{M}{2}\log_2\left(1+\frac{\frac{P}{M}\left(KN^2\right)^2}{e^2\frac{PK}{M}\mathrm{E}\left[\mathrm{tr}\left(\left(\mathbf{H}_k^H\mathbf{H}_k\right)^2\right)\left(\mathbf{G}_k\mathbf{G}_k^H\right)_{m,m}\right]+\sigma_1^2K\mathrm{E}\left[\left(\mathbf{H}_k^H\mathbf{H}_k(\mathbf{G}_k\mathbf{G}_k^H)^2\right)_{m,m}\right]+\sigma_2^2\rho_{MF}^{-2}}\right)
\\=\frac{M}{2}\log_2\left(1+\frac{PK^2N}{(e^2P+\sigma_1^2)KM\frac{M+N}{N}+\frac{PM(M+N)+M^2}{QN}\sigma_2^2}\right),
\end{multline}
\begin{multline}\label{2}
C_{MF-ZF}\overset{w.p.}{\longrightarrow}\\
\frac{M}{2}\log_2\left(1+\frac{\frac{P}{M}\left(KN\right)^2}{e^2\frac{PK}{M}\mathrm{E}\left[\mathrm{tr}\left(\left(\mathbf{H}_k^H\mathbf{H}_k\right)^2\left(\mathbf{G}_k\mathbf{G}_k^H\right)^{-1}\right)\right]+\sigma_1^2K\mathrm{E}\left[\left(\mathbf{H}_k^H\mathbf{H}_k\right)_{m,m}\right]+\sigma_2^2\rho_{MF-ZF}^{-2}}\right)
\\=\frac{M}{2}\log_2\left(1+\frac{PK^2N}{e^2PK\frac{M(M+N)}{N-M}+KM\sigma_1^2+\frac{PM(M+N)+M^2}{Q(N-M)}\sigma_2^2}\right),
\end{multline}
and
\begin{multline}\label{1}
C_{MF-RZF}\overset{w.p.}{\longrightarrow}
\frac{M}{2}\log_2\left(1+\right.\\\left.\frac{\frac{P}{M}\left(KN\mathrm{E}\left[\frac{\lambda}{\lambda+\alpha}\right]\right)^2}{e^2\frac{PK}{M}\mathrm{E}\left[\mathrm{tr}\left(\left(\mathbf{H}_k^H\mathbf{H}_k\right)^2\mathbf{G}_k^\alpha(\mathbf{I}_M-\alpha\mathbf{G}_k^\alpha)\right)\right]+\sigma_1^2K\mathrm{E}\left[\left(\mathbf{H}_k^H\mathbf{H}_k\left((\mathbf{I}_M-\alpha\mathbf{G}_k^\alpha)\right)^2\right)_{m,m}\right]+\sigma_2^2\rho_{MF-RZF}^{-2}}\right)
\\=\frac{M}{2}\log_2\left(1+\frac{PK^2N\left(\mathrm{E}[\frac{\lambda}{\lambda+\alpha}]\right)^2}{e^2PKM(M+N)\mathrm{E}[\frac{\lambda}{(\lambda+\alpha)^2}]+KM\mathrm{E}[\frac{\lambda^2}{(\lambda+\alpha)^2}]\sigma_1^2+\frac{PM(M+N)+M^2}{Q}\mathrm{E}[\frac{\lambda}{(\lambda+\alpha)^2}]\sigma_2^2}\right).
\end{multline}
\hrulefill
\end{figure*}
From the asymptotic capacities, we see that they satisfies the
scaling law in~\cite{4}, i.e., $C=(M/2) \log(K)+O(1)$ for large $K$.
Obviously the capacities will increase as $K$ increases, and derease
when $e$ increases. In addition, the CEG-noise power in $C_{MF-ZF}$
is the largest among those in the three asymptotic capacities,
resulting in a worse capacity performance of MF-ZF, which will be
confirmed by simulations.

From (\ref{7}), (\ref{2}), and (\ref{1}) it is observed that when
QNR (${Q}/{\sigma_2^2}$) grows to infinite for a fixed PNR
(${P}/{\sigma_1^2}$), the capacities of the three beamformers will
reach a limit, which demonstrates the "ceiling effect" that will be
confirmed by simulations. When PNR (=QNR) grows to infinite, the
capacities will grow linearly with PNR (dB) for perfect
$\mathcal{R}$-$\mathcal{D}$ CSI, or reach a limit for imperfect
$\mathcal{R}$-$\mathcal{D}$ CSI, which also demonstrates the
"ceiling effect" that will be confirmed by simulations.

Consider the $\mathcal{R}$-$\mathcal{D}$ CSI error varying with the
number of relays ($K$). Let $e=\sigma_q+K\sigma_d$, where $\sigma_q$
denotes the quantization error due to limited bits of feedback, and
$\sigma_d$ denotes the error weight caused by feedback delay in each
relay. Substituting such $e$ into the asymptotic capacities, it will
generate terms $O(\frac{1}{K}) + O(K)$ in the denominators of the
asymptotic SNR, which implies that the denominator will reach a
minimum value at some $K$. Therefore, there exists an optimal number
of relays to maximize the asymptotic capacities, which will be
confirmed by simulations.

Note that (\ref{2}) holds when $N>M$, because
$\mathrm{E}\left[\mathrm{tr}(\mathbf{G}\mathbf{G}^H)^{-1}\right]=\infty$  and
$\mathrm{E}\left[\rho_{ZF}^{-2}\right]=\infty$ for $M=N$~\cite{25}.
A bad capacity performance of MF-ZF can be expected due to the
infinite expectation of CEG-noise power  when $M=N$, which will be
confirmed by simulations. In order to optimize the regularizing
factor $\alpha$,  we shall use the following approximations.
\begin{equation}
\mathrm{E}\left[\frac{\lambda}{\lambda+\alpha}\right]=
\frac{1}{KM}\Sigma_{k=1}^K\Sigma_{m=1}^M\frac{\lambda_{m,k}}{\lambda_{m,k}+\alpha},\nonumber
\end{equation}
\begin{equation}
\mathrm{E}\left[\frac{\lambda}{\left(\lambda+\alpha\right)^2}\right]=
\frac{1}{KM}\Sigma_{k=1}^K\Sigma_{m=1}^M\frac{\lambda_{m,k}}{(\lambda_{m,k}+\alpha)^2}, \nonumber
\end{equation}
and
\begin{equation}
\mathrm{E}\left[\frac{\lambda^2}{\left(\lambda+\alpha\right)^2}\right]=
 \frac{1}{KM}\Sigma_{k=1}^K\Sigma_{m=1}^M\frac{\lambda_{m,k}^2}{(\lambda_{m,k}+\alpha)^2}, \nonumber
\end{equation}
 where $\lambda_{m,k}$ denotes the $m$th eigenvalue of $\mathbf{G}_k\mathbf{G}_k^H$. Take derivative of (\ref{1}) with respect to $\alpha$, and manipulate as~\cite{25}. Then we get the optimal regularizing factor as
\begin{equation}
\alpha_{\mathrm{opt}}=\frac{\frac{P(M+N)+M}{Q}\sigma_2^2+e^2PK(M+N)}{K\sigma_1^2}.
\end{equation}

\section{Simulation Results}
In this section, numerical results are carried out to validate what we draw from the analysis in the previous sections for the three relay beamforming schemes. The advantage of the optimized MF-RZF beamformer is also demonstrated.
\begin{figure}[!t]
\centering
\includegraphics[width=3.5in]{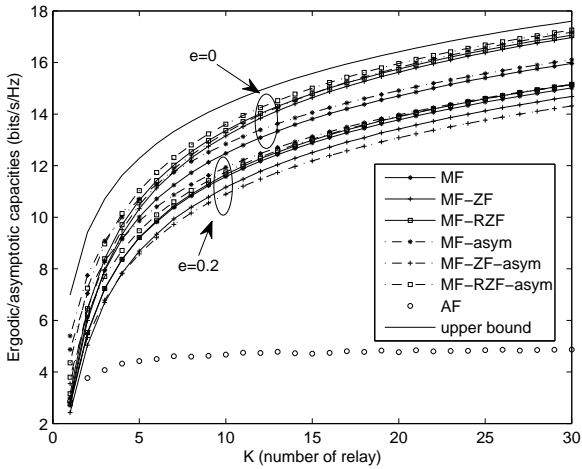}
\caption{Ergodic/asymptotic capacity vs. $K$ (number of relays) ($M=4$, $N=6$,
PNR=QNR=$10dB$). In this figure, MF-RZF is fixed with $\alpha=0.5$.}
\end{figure}

\subsection{Capacity Versus Number of Relays}
In Fig.~2, we compare the ergodic/asymptotic capacities of the three
beamforming schemes at relays.  We consider perfect CSIs of
$\mathcal{S}$-$\mathcal{R}$ channel at relays and imperfect CSI of
$\mathcal{R}$-$\mathcal{D}$ at relays. The solid curves are ergodic
capacities, and the dashed curves are the asymptotic capacities.
Capacities versus $K$ is demonstrated when $M=4$, $N=6$ and
PNR=QNR=$10dB$. When CSIs of $\mathcal{R}$-$\mathcal{D}$ are perfect
at relays ($e=0$), MF-RZF is the best choice.
When CSI of $\mathcal{R}$-$\mathcal{D}$ is imperfect at relays,
MF-ZF has apparently the worst performance. MF and MF-RZF have
almost the same performance. The  poor performance of MF-ZF under
imperfect CSI comes from the inverse Wishart distribution term in
its CEG-noise power, which can be clearly seen from the asymptotic
capacity (\ref{2}). Note that the asymptotic capacity of MF-ZF is
not tight enough to its ergodic capacity, since its power control
factor has an inverse Wishart distribution in the denominator. Thus,
the dynamic power control factor of MF-ZF has a much larger variance
than those of MF and MF-RZF. Since we use an average power control
factor instead of a dynamic power control factor to derive the
asymptotic capacity, it results in a gap between the two types of
capacities. We find that  the ergodic capacities still satisfy the
scaling law in~\cite{4}, i.e., $C=(M/2) \log(K)+O(1)$ for large $K$
in the presence of $\mathcal{R}$-$\mathcal{D}$ CSI error. This is
also consistent with the  asymptotic capacities for the three
beamformers. Note that AF keeps as the worst relaying strategy,
which  cannot utilize the distributed array gain.
\begin{figure}[!t]
\centering
\includegraphics[width=3.5in]{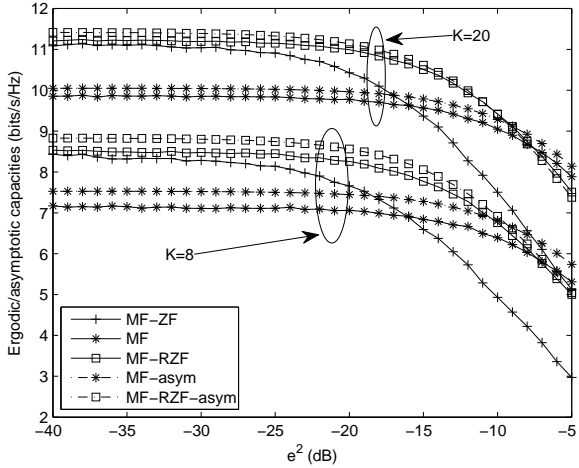}
\caption{Capacity vs. the power of $\mathcal{R}$-$\mathcal{D}$ CSI error ($N=M=4$, PNR=$5dB$,
QNR=$20dB$). Asymptotic capacities match with ergodic capacities very well for $K=20$. Since $\mathrm{E}\left[\mathrm{tr}(\mathbf{G}\mathbf{G}^H)^{-1}\right]=\infty$ for $M=N$, asymptotic capacity for MF-ZF does not exist. MF-RZF is fixed with $\alpha=0.5$.}
\end{figure}
\begin{figure}[!t]
\centering
\includegraphics[width=3.5in]{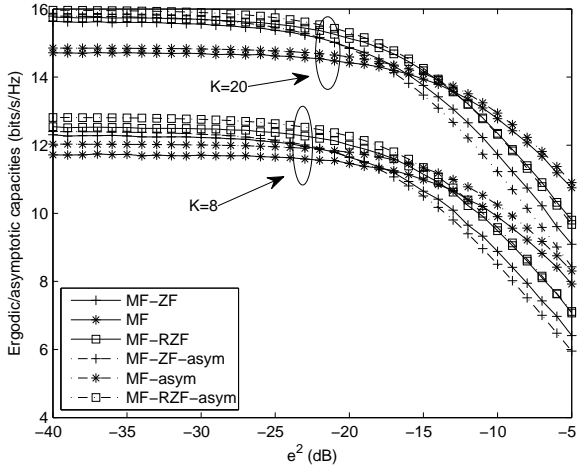}
\caption{Capacity vs. the power of $\mathcal{R}$-$\mathcal{D}$ CSI
error ($M=4$, $N=6$, PNR=$10dB$, QNR=$10dB$). MF-RZF is fixed with $\alpha=0.5$.}
\end{figure}
\begin{figure}[!t]
\centering
\includegraphics[width=3.5in]{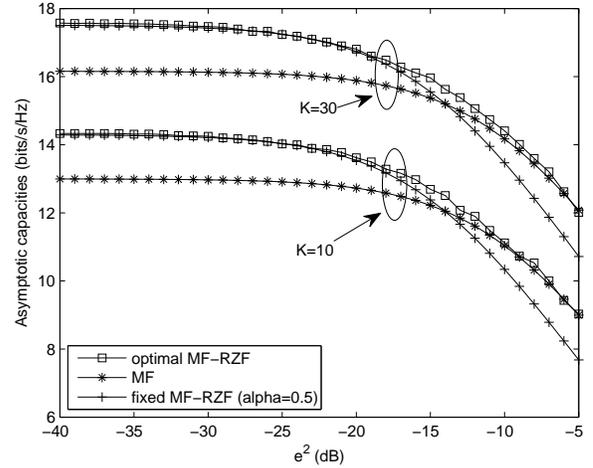}
\caption{Capacity vs. the power of $\mathcal{R}$-$\mathcal{D}$ CSI
error for MF and different MF-RZF ($M=4$, $N=6$, PNR=$10dB$,
QNR=$20dB$). The optimized MF-RZF outperforms MF for any power  of
$\mathcal{R}$-$\mathcal{D}$ CSI error.}
\end{figure}
\subsection{Capacity Versus Power of CSI Error}
In Fig.~3, we show the ergodic/asymptotic capacities versus power of
$\mathcal{R}$-$\mathcal{D}$ CSI error for $M=N=4$. The ergodic
capacity of MF-ZF drops quickly when CSI error occurs, which
validates what we observed in (\ref{2}).  We see that the asymptotic
capacities of MF and MF-RZF match well with their ergodic capacities
for  $K=20$, which also shows that static power allocation has
almost the same performance as dynamic power allocation for large
$K$. Similarly, Fig.~4 is the case  for $M=4$ and $N=6$.  It is
observed that when $N>M$, MF-ZF and MF-RZF obviously outperforms MF
with perfect $\mathcal{R}$-$\mathcal{D}$ CSI at relays, while their
performance will be upside down in the presence of
$\mathcal{R}$-$\mathcal{D}$ CSI error. Since ergodic capacities
match well with the asymptotic capacities, in the rest figures, we
only plot the asymptotic capacities to demonstrate the advantage of
optimized MF-RZF. Fig.~5 shows that the optimized MF-RZF has
consistently the best performance for any powers of  CSI error.
\begin{figure}[!t]
\centering
\includegraphics[width=3.5in]{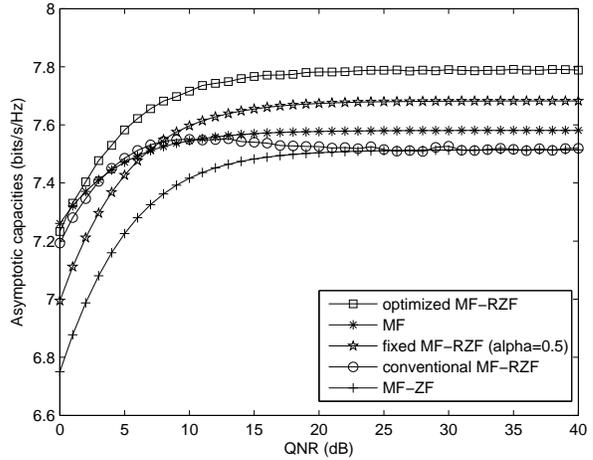}
\caption{Capacity vs. QNR ($M=2$, $N=4$, PNR=$10dB$,
$e=0.1$). Optimized MF-RZF outperforms the MF-RZF with fixed $\alpha$, MF, the conventional MF-RZF and MF-ZF in the presence of  $\mathcal{R}$-$\mathcal{D}$ CSI error. All schemes experience the ceiling effect.}
\end{figure}
\begin{figure}[!t]
\centering
\includegraphics[width=3.5in]{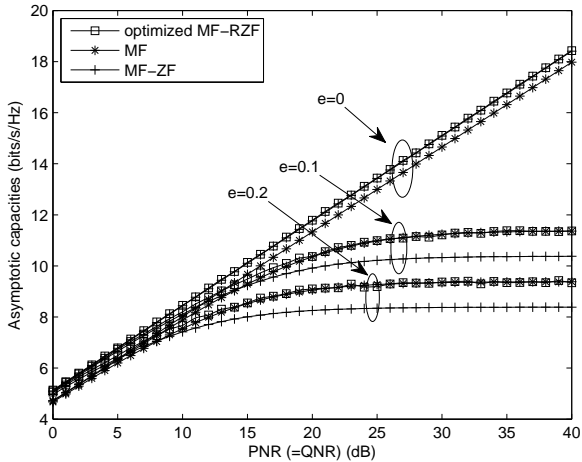}
\caption{Capacity vs. PNR (=QNR ) for
$e=0, 0.1, 0.2$. All schemes experience the ceiling effect in the presence of $\mathcal{R}$-$\mathcal{D}$ CSI error.}
\end{figure}
\begin{figure}[!t]
\centering
\includegraphics[width=3.5in]{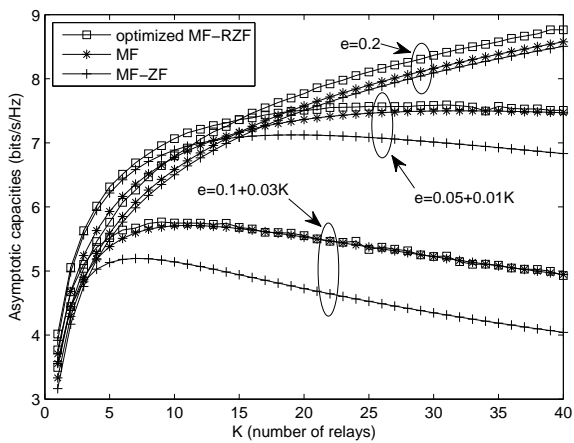}
\caption{Capacity vs. $K$ for $e=0.2$ and $e=\sigma_q+K\sigma_d$.
Each scheme has an optimal number of relays to maximize the capacity
in the presence of CSI error. The optimal point of $K$ is obvious if
$\sigma_d$ is big.}
\end{figure}
\subsection{Capacity Versus PNR and QNR}
A more apparent superiority of optimized MF-RZF can be observed in
Fig.~6, where we fix the SNR of $\mathcal{S}$-$\mathcal{R}$ channel
(PNR) and increase the SNR of $\mathcal{R}$-$\mathcal{D}$ channel
(QNR). Note that the capacity in this scenario is limited by the
$\mathcal{S}$-$\mathcal{R}$ channels~\cite{4}, called ``ceiling
phenomenon'' in~\cite{27}. We also include the conventional
optimized RZF ($\alpha=M\sigma_2^2/Q=M/QNR$) for perfect
CSI~\cite{25} in Fig.~6, and refer to it as the conventional MF-RZF.
We find that the optimized MF-RZF outperforms MF, MF-ZF, the
conventional MF-RZF and MF-RZF with fixed $\alpha$, and has the
highest ceiling for QNR $>1$dB.  When QNR increases, the $\alpha$ in
the conventional MF-RZF approaches to zeros. So MF-RZF will converge
to MF-ZF beamformer. In Fig.~7, we increase PNR and QNR
simultaneously. When $\mathcal{R}$-$\mathcal{D}$ CSI is perfect at
relays, the capacities of all the three beamformings grow linearly
with the PNR (=QNR) in dB. When $\mathcal{R}$-$\mathcal{D}$ CSI
error occurs, we see different capacity limits for different CSI
error powers. This is the "ceiling effect" discussed in Section
\uppercase \expandafter {\romannumeral 5}.

\subsection{Capacity Versus Relay Number for Dynamic CSI Error}

Fig.~8 shows the capacities versus the relay number $K$ when CSI
error $e=\sigma_q+K\sigma_d$. We also include the capacities versus
relay number $K$ with constant CSI error for comparison in Fig.~8.
It is observed that the capacity achieves maximum at some optimal
relay number in the presence of CSI error. When $\sigma_d$ is
bigger, the optimal $K$ is smaller.


\section{Conclusion}
In this paper, considering three efficient relay beamforming schemes
based on MF, MF-ZF and MF-RZF techniques, we investigate the effect
of imperfect $\mathcal{R}$-$\mathcal{D}$ CSI at relays to the capacity in a dual-hop MIMO multi-relay
network with Amplify-and-Forward (AF) relaying protocol. Supposing
Gaussian distributed $\mathcal{R}$-$\mathcal{D}$ CSI error and
perfect $\mathcal{S}$-$\mathcal{R}$ CSI at relays, we give the
ergodic capacities of  the three beamformers in terms of
instantaneous SNR. Using Law of Large Number, we derive the
asymptotic capacities of the three beamformers for large number of
relays, upon which, the optimized MF-RZF is derived. Simulation
results show that the asymptotic capacities match with the
respective ergodic capacities very well. Analysis and simulations
demonstrate that MF-ZF beamformer has the worst performance in the
presence of $\mathcal{R}$-$\mathcal{D}$ CSI error. The capacity of
MF-RZF drops faster than that of MF as the
$\mathcal{R}$-$\mathcal{D}$ CSI error increases, and has a small
performance loss compared to that of MF, while the optimized MF-RZF
has consistently the best performance for any power of
$\mathcal{R}$-$\mathcal{D}$ CSI error. Although we consider imperfect $\mathcal{R}$-$\mathcal{D}$ CSI caused by limited feedback and large delay, imperfect $\mathcal{S}$-$\mathcal{R}$ CSI is still a practical consideration when estimation error presents and we will consider this case as our future work.


\begin{thebibliography}{1}

\bibitem{1}
E. Telatar, ``Capacity of Multi-antenna Gaussian Channels,"
\emph{Euro. Trans. Telecomm.}, vol. 10, no. 6, pp. 585-596, Nov.
1999.

\bibitem{2}
A. Goldsmith, S. A. Jafar, N. Jindal, and S. Vishwanath, ``Capacity
limits of MIMO channels," \emph{IEEE J. Sel. Areas Commun.}, vol.
51, no. 6, pp. 684-702, Jun. 2003

\bibitem{3}
B. Wang, J. Zhang, and A. H. Madsen, ``On the Capacity of MIMO Relay
Channels," \emph{IEEE Trans. Inf. Theory}, vol. 51, no. 1, pp.
29-43, Jan. 2005.

\bibitem{5}
C. Chae, T. Tang, R. W. Heath, Jr., and S. Cho, ``MIMO Relaying With
Linear Processing for Multiuser Transmission in Fixed Relay
Networks," \emph{IEEE Trans. on Signal Process.}, vol. 56, no. 2,
pp. 727-738, Feb. 2008.

\bibitem{7}
R. H. Y. Louie, Y. Li, and B. Vucetic, ``Performance analysis of
beamforming in two hop amplify and forward relay networks," in
\emph{IEEE Intern. Conf. Commun. (ICC)}, pp. 4311-4315, May 2008.

\bibitem{31}
F. Tseng and W. Wu, ``Nonlinear transceiver designs in MIMO amplity-and forward relay systems,"\emph{IEEE Trans. Vehicular. Technology}, vol. 60, no. 2, pp. 528-538, February. 2011.
\bibitem{32}
F. Tseng and W. Wu, ``Linear MMSE transceiver design in amplity-and forward MIMO relay systems,"\emph{IEEE Trans. Vehicular. Technology}, vol. 59, no. 2, pp. 754-765, February. 2011.
\bibitem{33}
M. Peng, H. Liu, W. Wang and H. Chen, ``Cooperative network coding with MIMO transmission in wireless decode-and-forward relay networks,"\emph{IEEE Trans. Vehicular. Technology}, vol. 59, no. 7, pp. 3577-3588, Sep. 2010.

\bibitem{4}
H. Bolcskei, R. U. Nabar, O. Oyman, and A. J. Paulraj, ``Capacity
Scaling Laws in MIMO Relay Networks," \emph{IEEE Trans. on Wireless
Commun.}, vol. 5, no. 6, pp. 1433-1444, June 2006.



\bibitem{6}
W. Guan, H. Luo, and W. Chen, ``Linear Relaying Scheme for MIMO
Relay System With QoS Requirements," \emph{IEEE Signal Process.
Lett.}, vol. 15, pp. 697-700, 2008.




\bibitem{8}
H. Shi, T. Abe, T. Asai, and H. Yoshino, ``A relaying scheme using QR
decomposition with phase control for MIMO wireless networks," in
\emph{Proc. Int. Conf. on Commun.}, May 2005, vol. 4, pp. 2705-2711.

\bibitem{9}
H. Shi, T. Abe, T. Asai, and H. Yoshino, ``Relaying schemes using
matrix triangularization for MIMO wireless networks," \emph{IEEE
Trans. Commun.}, vol. 55, pp. 1683-1688, Sep. 2007.

\bibitem{40}
B. Zhang, Z. -Q. He, K. Niu, and L. Zhang, ''Robust Linear Beamforming for MIMO Relay Broadcast Channel with Limited Feedback'' \emph{IEEE Signal Process. Lett.}, vol. 17, no. 2, pp. 209-212, February. 2010.



\bibitem{14}
T. Yoo, A. Goldsmith, ``Capacity and power allocation for fading MIMO
channels with channel estimation error," \emph{IEEE Trans. Inf.
Theory.}, vol. 52, no. 5, pp. 2203-2214, May. 2006.

\bibitem{16}
E. Martos, J. F. Paris, U. Fernandez, and A. J. Goldsmith, ``Exact
BER analysis for M-QAM modulation with transmit beamforming under
channel prediction errors," \emph{IEEE Trans. Wireless Commun.},
vol. 7, no. 10, pp. 3674-3678, Oct. 2008.

\bibitem{19}
T. Weber, A. Sklavos, and M. Meurer, ``Imperfect channel-state
information in MIMO transmission," \emph{IEEE Trans. Commun.}, vol.
54, no. 3, pp. 543-552, Mar. 2006.

\bibitem{15}
X. Zhou, P. Sadephi, T. A. Lamahewa, and S. Durrani, ``Optimizing
antenna configuration for MIMO systems with imperfect channel
estimation," \emph{IEEE Trans. Wireless Commun.}, vol. 8, no. 3, pp. 1177-1181,
Mar. 2009.



\bibitem{17}
M. Payaro, A. Pascual, A. I. Perez, and M. A. Lagunas, ``Robust
design of spatial Tomlinson-Harashima precoding in the presence of
errors in the CSI," \emph{IEEE Trans. Wireless Commun.}, vol. 6, no.
7, pp. 2396-2401, Jul. 2007.

\bibitem{18}
R. Mo, J. Lin, Y. H. Chew, and W. H. Chin, ``Relay precoder design
for non-regenerative MIMO relay networks with imperfect channel
state information," \emph{Proc. Int. Conf. on Commun.}, May. 2010.



\bibitem{20}
E. Baccarelli, M. Biagi, and C. Pelizzoni, ``On the information
throughput and optimized power allacation for MIMO wireless systems
with imperfect channel estimation," \emph{IEEE Trans. Signal.
Process.}, vol. 53, no. 7, pp. 2335-2347, Jul. 2005.



%


\bibitem{21}
B. Hassibi, and B. M. Hochwald, ``How much training is needed in
multiple-antenna wireless links?" \emph{IEEE Tran. Inf. Theory.},
vol. 49, no. 4, pp. 951-963, Apr. 2003.
\bibitem{34}
A. Y. Panah and R. W. Health, ``MIMO two-way amplify-and-forward relaying with imperfect receiver CSI,"\emph{IEEE Trans. Vehicular. Technology}, vol. 59, no. 9, pp. 4377-4387, Nov. 2010.
\bibitem{10}
Y. Zhang, H. Luo, and W. Chen, ``Efficient relay beamforming design
with SIC  detection for Dual-Hop MIMO relay networks,"
\emph{IEEE Trans. Vehicular. Technology}, vol. 59, no. 8, pp. 4192-4197, 2010.



\bibitem{13}
J. K. Zhang, A. Kavcic, and K. M. Wong, ``Equal-diagonal QR
decomposition and its application to precoder design for
successive-cancellation detection," \emph{IEEE Trans. Inf. Theory},
vol. 51, no. 1, pp. 154-172, Jan. 2005.

\bibitem{27}
A. D. Dabbagh and D. J. Love, ``Multiple Antenna MMSE Based Downlink Precoding with Quantized Feedback or Channel Mismatch" \emph{IEEE Trans. Commun.}, vol. 56, no. 11, pp. 1859-1868, November. 2008.

\bibitem{22}
C. Wang, E. K. S. Au, R. D. Murch, W. H. Mow, R. S. Cheng, and V.
Lau, ``On the performance of the MIMO Zero-forcing receiver in the
presence of channel estimation error," \emph{IEEE Trans. Wireless
commun.}, vol. 6, no. 3, pp. 805-810, Mar. 2007.




\bibitem{25}
C. Peel, B. Hochwald, and A. Swindlehurst, ``Vector-perturbation
technique for near-capacity multiantenna multiuser
communication-Part I: Channel inversion and regularization,"
\emph{IEEE Trans. Commun.}, vol. 53, no. 1, pp. 195-202, Jan. 2005.



\bibitem{28}
A. Lozano, A. M. Tulino, and S. Verdu, ``Multiple-antenna capacity in the low-power regime," \emph{IEEE Trans. Inf. Theory}, vol. 49, no. 10, pp. 2527-2544, Oct. 2003.


\bibitem{30}
A. M. Tulino and S. Verdu, ``Random matrix theory and wireless communications," \emph{Foundations and Trens in Communications and Information Theory}, vol. 1, no. 1, pp. 1-182, 2004.




\end{thebibliography}
\end{document}